\documentclass[conference,letterpaper]{IEEEtran}

\usepackage[utf8]{inputenc} 
\usepackage[cmex10]{amsmath}
\usepackage[T1]{fontenc}
\usepackage{amsthm}
\usepackage{url}
\usepackage{ifthen}
\usepackage{cite}
\usepackage{physics}
\usepackage{graphicx}

\newtheorem{definition}{Definition}

\newtheorem{lemma}{Lemma}

\usepackage{enumerate}
\usepackage{verbatim}
\usepackage{enumerate}
\usepackage{framed,xspace}
\usepackage{epsfig}
\usepackage{setspace}
\usepackage{bbm}
\usepackage{dsfont}
\usepackage{ifthen}
\usepackage{xcolor}
\usepackage{stix}
\usepackage[T1]{fontenc}
\usepackage{nicefrac}
\usepackage{enumitem}
\usepackage{thmtools}
\usepackage{thm-restate}
\usepackage{hyperref}
\usepackage{cleveref}

\newcommand{\eps}{\varepsilon}

\newcommand{\mX}{\mathcal X}
\newcommand{\mY}{\mathcal Y}

\newcommand{\mA}{\mathcal A}

\newcommand{\mS}{\mathcal S}
\newcommand{\mZ}{\mathcal Z}

\newcommand{\mE}{\mathcal E}

\newcommand{\K}[2]{mss({#1},{#2})}

\def\01{\{0,1\}}

\newcommand{\OT}[3]{{\binom{#2}{#1}}{\textsf{-OT}^{#3}}}

\newcommand{\RabinOT}[2]{{({#1})\textsf{-RabinOT}^{#2}}}

\newcommand*{\sbin}{\{0,1\}}

\DeclareMathOperator{\dis}{D}

\newcommand{\IP}[1]{{\textsf{IP}_{#1}}}

\newcommand{\EQ}[1]{{\textsf{EQ}_{#1}}}

\newcommand{\cHS}[2]{{\Hop(#1|#2)}}

\newcommand{\Hop}{H}
\newcommand{\HS}[2]{{\Hop(#1|#2)}}

\newcommand{\Hmaxeps}[2]{H_{0}^{#1}(#2)}
\newcommand{\Chmaxeps}[3]{H_{0}^{#1}(#2|#3)}

\newcommand{\support}{\textnormal{supp}}

\newcommand{\chh}[5]{\ensuremath{H_{#1}^{#2}({#3}|{#4})_{#5}}}
\newcommand{\chvn}[3]{\chh{}{}{#1}{#2}{#3}}
\newcommand{\proj}[1]{|#1\rangle\langle#1|}

\newcommand{\supp}[1]{\textnormal{supp}\, ( #1 )}
\newcommand{\id}{{\mathbbm{1}}}

\newcommand{\vecstate}[1]{\ket{#1}\bra{#1}}

\newcommand{\h}{\ensuremath{\mathcal{H}}}
\newcommand{\hi}[1]{\ensuremath{\mathcal{H}_{\textnormal{#1}}}}
\newcommand{\hA}{\hi{A}}
\newcommand{\hX}{\hi{X}}
\newcommand{\hAB}{\hi{AB}}
\newcommand{\hB}{\hi{B}}

\newcommand{\hR}{\hi{R}}

\newcommand{\idx}[2]{{#1}_{#2}}

\newcommand{\rhoA}{\ensuremath{\idx{\rho}{A}}}

\newcommand{\rhoAB}{\ensuremath{\idx{\rho}{AB}}}
\newcommand{\rhoABZ}{\ensuremath{\idx{\rho}{ABZ}}}

\newcommand{\rhoABC}{\ensuremath{\idx{\rho}{ABC}}}

\newcommand{\psiAB}{\ensuremath{\idx{\psi}{AB}}}

\newcommand{\kron}{\otimes}
\newcommand{\ptrace}[2]{\ensuremath{\ptr{#1} (#2)}}
\newcommand{\idi}[1]{\ensuremath{\mathds{1}_{\textnormal{\tiny #1}}}}
\newcommand{\idA}{\idi{A}}

\newcommand{\idAB}{\idi{AB}}

\newcommand{\ptr}[1]{\operatorname{tr}_{\textnormal{#1}}}

\newcommand{\posops}[1]{\ensuremath{\mathcal{P}(#1)}}

\newcommand{\states}[1]{\ensuremath{\mathcal{S}(#1)}}

\newcounter{protoCount}
\newcounter{protoList}
\newsavebox{\tmpbox}
\newlength{\protobox}

\newcounter{prot}
\newtheorem{condition}{Condition}{\bfseries}{}

\begin{document}
\title{Lower Bounds for Quantum Secure Function Evaluation Reductions } 

% %%% Single author, or several authors with same affiliation:
% \author{%
%  \IEEEauthorblockN{Author 1 and Author 2}
% \IEEEauthorblockA{Department of Statistics and Data Science\\
%                    University 1\\
 %                   City 1\\
  %                  Email: author1@university1.edu}% }

%%% Several authors with up to three affiliations:
\author{%
  \IEEEauthorblockN{Esther H\"anggi}
  \IEEEauthorblockA{%Lucerne School of Computer Science and Information Technology \\
                    Lucerne University of Applied Sciences and Arts\\
                    Rotkreuz\\
                    Email: esther.haenggi@hslu.ch}
  \and
  \IEEEauthorblockN{Severin Winkler}
  \IEEEauthorblockA{Ergon Informatik AG\\ 
                    Zurich\\
                    Email: severin.winkler@ergon.ch}
}

\maketitle

%%%%%%
%% Abstract: 
%% If your paper is eligible for the student paper award, please add
%% the comment "THIS PAPER IS ELIGIBLE FOR THE STUDENT PAPER
%% AWARD." as a first line in the abstract. 
%% For the final version of the accepted paper, please do not forget
%% to remove this comment!
%%

\begin{abstract}
One-sided output secure function evaluation is a cryptographic primitive where the two mutually distrustful players, Alice and Bob, both have a private input to a bivariate function. Bob obtains the value of the function for the given inputs, while Alice receives no output. It is known that this primitive cannot be securely implemented if the two players only have access to noiseless classical and quantum communication. In this work, we first show that Bob can extract the function values for \emph{all} his possible inputs from any implementation of a non-trivial function that is correct and preserves the privacy of Bob's input. Our result holds in the non-asymptotic setting where the players have finite resources and the error is a constant. Then we consider protocols for secure function evaluation in a setup where the two players have access to trusted distributed randomness as a resource. Building upon the first result, we prove a bound on the efficiency of such cryptographic reductions for any non-trivial function in terms of the conditional entropies of the trusted randomness. From this result, we can derive lower bounds on the number of instances of different variants of OT needed to securely implement a given function. 
\end{abstract}

\section{Introduction}
Secure multi-party computation enables two or more mistrustful parties to collaborate in order to achieve a common goal~\cite{Yao82}. Practical tasks are, for example, electronic voting or online auctions. A specific variant of multi-party computation is \emph{secure function evaluation}, where each party has a private input and the value of the function is computed correctly and securely, meaning that no party can cheat or learn more than what is implied by their own input and output. In the case where two parties, Alice and Bob, compute a bivariate function and only Bob receives the computed value of the function, the task is known as \emph{one-sided output secure function evaluation} (SFE). 

An important and well-studied variant of SFE is one-out-of-two bit-OT ($\OT{1}{2}{1}$)~\cite{EvGoLe85}. Here, the sender (Alice) has two input bits, $x_0$ and $x_1$. The receiver (Bob) gives as input a choice bit $c$ and receives $x_c$ without learning $x_{1-c}$. The sender gets no information about the choice bit $c$. This primitive can be generalized to t-out-of-n-OT ($\OT{t}{n}{k}$) where the $n$ inputs are strings of $k$ bits and the receiver can choose to receive $t<n$ of them. OT is sufficient to execute any two-party computation securely \cite{GolVai87,Kilian88}. OT can be precomputed offline~\cite{BBCS92,Beaver95}, which means that the players compute distributed randomness by using OT with random inputs and storing the inputs and outputs. The stored shared randomness can then be used as OT at a later stage, with a noiseless communication channel only.

SFE is impossible to implement with information-theoretic security, even if the players can manipulate quantum information and have access to noiseless classical and quantum communication channels~\cite{Lo97}.

Since SFE cannot be implemented from scratch, 
there has been a lot of interest in reductions of secure function evaluation to weaker resource primitives. OT --- and therefore ultimately any secure two-party computation --- can be realized from noisy channels \cite{CreKil88,CrMoWo04,DFMS04,Wullsc09}, noisy correlations~\cite{WolWul04,WolWul08,NaWi06}, or weak variants of oblivious transfer~ \cite{CreKil88,Cachin98,DaKiSa99,BrCrWo03,DFSS06,Wullsc07}. 

Given these positive results, it is natural to ask how efficient such reductions can be in principle, i.e., how many instances of a primitive are needed to securely implement a certain function. Several lower bounds for OT reductions are known. The earliest impossibility result~\cite{Beaver89a} shows that the number of $\OT{1}{2}{1}$ cannot be \emph{extended}, i.e., there does not exist a protocol using $n$ instances of $\OT{1}{2}{1}$ that perfectly implements $m > n$ instances. Lower bounds on the number of instances of OT needed to implement other variants of OT~\cite{DodMic99, Maurer99, KKK08,WolWul08, PrabhakaranP14} and SFE~\cite{WW14} are also known.

Characterizing the efficiency of reductions to resource primitives leads to a better understanding of the difficulty of the implemented task and of the \emph{cryptographic capacity} of the resource. In this spirit, the OT capacity of noisy resources has been introduced in~\cite{NW08}. It is the maximum length of string-OT per instance of the resource that can be securely implemented. Both lower bounds~\cite{NW08,Ahlswede2013, RP14} and upper bounds~\cite{Ahlswede2013,SudaWY24} on the OT capacity of noisy resources are known.

All these bounds on the efficiency of OT and general SFE reductions only hold for classical protocols. In this work, we, therefore, study the efficiency of such reductions in the setting where the two players can manipulate and exchange quantum systems.

\subsection{Previous Results}

The authors of~\cite{SaScSo09} considered quantum protocols where the players have access to resource primitives, including different variants of OT. They proved that important lower bounds for classical protocols also apply to \emph{perfectly} secure quantum reductions. In~\cite{WW10} it has been demonstrated that \emph{statistically} secure protocols can violate these bounds by an arbitrarily large factor. More precisely, there exists a quantum protocol that reverses string OT much more efficiently than any classical protocol~\cite{Ahlswede2013,Maurer99,DodMic99,WolWul08,WW10,PrabhakaranP14,HimWat15}. While this proves that quantum reductions are more powerful than classical protocols, a weaker lower bound still holds for statistically secure quantum protocols for $\OT{1}{2}{k}$~\cite{WW10}. This bound enables a generalization of the result from~\cite{Beaver96}, showing that OT cannot be extended by quantum protocols. Using the equivalence of $\OT{1}{2}{k}$ and commitments in the quantum setting~\cite{BBCS92,Yao95,DFLSS09}, this also implies restrictions on protocols extending commitments. The general impossibility of extending bit commitments has been shown later in~\cite{WTHR11}. To the best of our knowledge, the only lower bounds for SFE in the statistical case apply only to $\OT{1}{2}{k}$ and cannot easily be adapted to general SFE.

\subsection{Our Results}
We consider statistically secure quantum protocols that compute a function between two parties. We show that for any non-trivial function and without any cryptographic resource primitive, such protocols are completely insecure because Bob can compute the function values for all his possible inputs. We give a quantitative bound in the non-asymptotic setting for the success probability of such an attack (Theorem~\ref{thm:impossibility-sfe}). 

For the main result of this work, we consider implementations of SFE from a shared cryptographic resource, represented by trusted randomness distributed to the players. Building upon the first result, we provide a bound on the efficiency of such reductions — in terms of the conditional entropies of the randomness. This allows us to derive bounds on the minimal number of OTs needed to compute a function securely. Our results hold in the non-asymptotic regime, i.e., we consider a finite number of resource primitives and our results hold for any error.

Our lower bounds implies the following results:
\begin{itemize}
\item Quantum protocols cannot extend $\OT{1}{n}{k}$. There exists a constant $c_n > 0$ such that any quantum reduction of $m + 1$ instances of $\OT{1}{n}{1}$ to $m$ instances of $\OT{1}{n}{1}$ must have an error of at least $\frac{c_n}{m}$ (Corollary~\ref{cor:extend-ot}). This result was previously only known for $\OT{1}{2}{1}$.
\item Implementing $\OT{1}{n}{k}$ from $\OT{1}{2}{k}$ needs approximately $n-1$ instances for large $k$ (Corollary~\ref{cor:extend-choices-ot}). Previously, only a lower bound of $1$ instance was known. There is a classical protocol~\cite{DodMic99} that achieves this bound and which is also secure against quantum adversaries\footnote{Using the quantum lifting theorem in~\cite{Unruh10} in the quantum UC model, the security against quantum adversaries is straightforward, but somewhat tedious to show.}. By our result this protocol is essentially optimal also in the quantum case. 
\item Furthermore, we derive bounds on the efficiency of secure implementations of well-known functions such as the inner-product-modulo-two function (Corollary~\ref{cor:reductions:inner-product}) and the equality function (Corollary~\ref{cor:reductions:equality}). We are not aware of any existing quantum bound for these functions. While quantum protocols for SFE can be arbitrarily more efficient in general~\cite{WW14}, these bounds show that they can only be slightly more efficient than optimal classical protocols in these two cases. 
\end{itemize}

\subsection{Outline}
The rest of this paper is organized as follows. In Section~\ref{sec:prel}, we introduce the notations and mathematical tools. In Section~\ref{sec:sec-definitions-sfe}, we define the cryptographic functionality of SFE and necessary conditions that any statistically secure implementation of SFE must fulfill. Then, in Section~\ref{sec:imp_sfe}, we state our first result showing the impossibility of implementing any non-trivial function with statistical security. Our main result, a bound on the efficiency of reductions of SFE to trusted distributions, is presented in Section~\ref{sec:reductions-of-sfe}. 

\section{Preliminaries}\label{sec:prel}

The distribution of a random variable $X$ is denoted by $P_X(x)$. Given the distribution $P_{XY}$ over $\mX \times \mY$, the marginal distribution is $P_{X}(x) := \sum_{y \in \mY} P_{XY}(x,y)$.

Quantum states $\states{\h}$ are represented by positive semi-definite operators of trace one acting on a Hilbert space $\h$, $\states{\h} := \{ \rho \in \posops{\h}:\h\rightarrow \h | \tr\,\rho = 1 \}$. We consider finite-dimensional Hilbert spaces. 
A state $\rho \in \states{\h}$ is \emph{pure} if it has rank one and \emph{mixed} otherwise. A pure state $\rho$ can be represented by an element $\ket{\psi}$ of the Hilbert space $\h$, with $\rho = \proj{\psi}$ the outer product. 
We use indices to denote multipartite Hilbert spaces, i.e., $\hAB:=\hA\kron\hB$. Given a quantum state $\rhoAB \in \states{\hA\kron\hB}$ we denote by $\rhoA$ its marginal states $\rhoA=\ptrace{B}{\rhoAB}$.

A (classical) probability distribution $P_X$ of a random variable $X$ over $\mX$ can be represented by a quantum state $\rho_X=\sum_{x\in\mX}P_X (x) \vecstate{x}$, using 
an orthonormal basis $\{\ket{x} \mid x\in\mX\}$ of the Hilbert space $\hi{$\mX$}$. The uniform distribution on $\mX$ corresponds to $\tau_{\mX}:=\frac{1}{|\mX|}\sum_{x\in\mX}\vecstate{x}$.
Quantum information about a classical random variable $X$ can be represented by a classical-quantum state or \emph{cq-state} $\rho_{XB}$ on $\hi{$\mX$}\kron\hi{B}$ of the form $\rho_{XB}=\sum_{x \in\mX}P_X(x)\vecstate{x}\kron\rho_B^x$. 

A mixed state can be seen as part of a pure state on a larger Hilbert space (see for example \cite{NieChu00}). More precisely, given $\rho_A$ on $\hA$, there exists a pure density operator $\rho_{AB}$ on a
joint system $\hA \kron \hB$ such that $\rho_A = \ptrace{B}{\rho_{AB}}$. $\rho_{AB}$ is called a \emph{purification} of $\rho_A$. For a state representing a classical distribution $\rho_X=\sum_{x\in\mX}P_X (x) \vecstate{x}$, one explicit purification is $\ket{\psi}_{XX'}=\sum_{x\in\mX} \sqrt{P_X (x)} \ket{x}_{X}\kron\ket{x}_{X'}$.

The trace distance measures the probability that two quantum states can be distinguished.
\begin{definition}[Trace Distance]
The \emph{trace-distance} between two quantum states $\rho$ and $\sigma$ is 
\begin{align}
D(\rho,\sigma)&= \frac{1}{2}\| \rho-\sigma \|_1\;,
\end{align} 
where $\|A\|_1=\tr \sqrt{A^\dagger A}$. Two states $\rho$ and $\sigma$ with trace-distance at most $\delta$ are called \emph{$\delta$-close} and we write $\rho \approx_{\delta} \sigma$. 
\end{definition}

Any transformation $\mE$ of a quantum system can be represented by a trace-preserving completely positive map (TP-CPM). According to~\cite{Stine55}, any TP-CPM has a Stinespring dilation, i.e., there exists a Hilbert space $\hR$, a unitary $U$ acting on $\h_{AXA'R}$ and a pure state $\sigma_{XA'R}\in \states{\h_{XA'R}}$ with $\mE(\rhoA)=\ptrace{AR}{U(\rhoA \otimes \sigma_{XA'R}) U^\dagger}$.

\subsection{Entropies}
We use the notation $h(p)=-p\log(p)-(1-p)\log(1-p)$
for the binary entropy function.

\begin{definition}
The \emph{conditional von Neumann entropy} of a quantum system $A$ given another quantum system $B$ described by a joint state $\rho_{AB}$ is 
\[ \chvn{A}{B}{\rho} := H(\rho_{AB}) - H(\rho_{B})\;,\]
where $H(\rho) := \tr( - \rho \log(\rho))$.
\end{definition}

For classical states defined by a probability distribution, we use the \emph{conditional max-entropy}. This entropy corresponds to the R\'enyi entropy \cite{Renyi61} of order 0. $\Hmaxeps{}{X}$ is defined as the logarithm of the size of the support of $X$. In the conditional case, we maximize the conditional distribution over all values of the conditioning random variable.
\begin{definition}[Max-Entropy]\label{def:class-max-entropy}
For random variables $X,Y$ with distribution $P_{XY}$ over $\mX \times \mY$, the conditional max-entropy is
\begin{align*}
 \Chmaxeps{}{X}{Y}&:=\max_{y \in \mY}\log|\supp{P_{X|Y=y}}|\;.
\end{align*}
\end{definition}

\section{Definitions of Secure Function Evaluation}\label{sec:sec-definitions-sfe}
In this section, we define the security of one-sided output secure function evaluation in the malicious model and derive necessary security conditions from this definition. First, we define the ideal functionality.
\begin{definition}[Ideal SFE]\label{def:ideal-sfe}
An \emph{ideal one-sided output secure function evaluation (SFE)} takes an input $x\in \mX$ on Alice's side and an input $y\in \mathcal{Y}$ on Bob's side. It outputs the function value $f(x,y)$ on Bob's side and nothing on Alice's side. (See Figure~\ref{fig:ideal-sfe}) 
\end{definition}
\begin{figure}[htbp]
   \centering
   \includegraphics[width=0.5\textwidth]{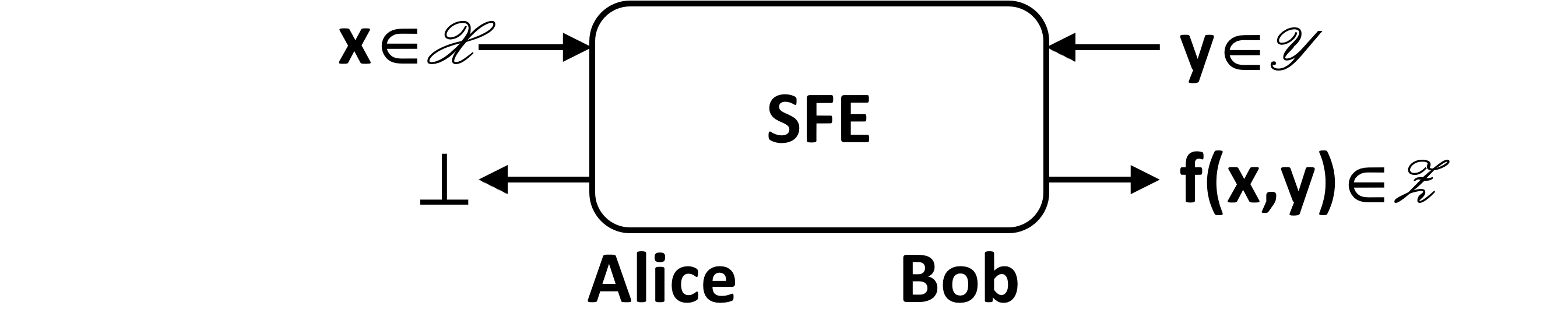}
   % where an .eps filename suffix will be assumed under latex,
%   % and a .pdf suffix will be assumed for pdflatex
   \caption{Ideal (one-sided output) SFE.}
  \label{fig:ideal-sfe}
 \end{figure}

Intuitively, an implementation of such a functionality is considered secure if a malicious player can learn nothing beyond what they could derive from their input and output in the ideal setting \cite{GoMiRa85,GoMiWi87,MicRog91,Beaver91}. This requirement is formally captured by the following definition.
\begin{definition}[Statistical Security of SFE]A protocol is an $\eps$-secure implementation of a function $f:\mX \times \mY \rightarrow \mZ$ in the malicious model if for any adversary $\mA$ attacking the protocol (real setting), there exists a \emph{simulator} $\mS$ using the ideal function (ideal setting) such that for all inputs of the honest players the real and the ideal setting can be distinguished with an advantage of at most $\eps$.
\end{definition}
This definition implies the following security conditions. These conditions are not sufficient to prove the security of a protocol, but any secure protocol must necessarily fulfill these conditions~\cite{FS09} (see also~\cite{WW14}). Since we are interested in the impossibility of protocols, this only strengthens our results.

\begin{condition}\label{con:secure-sfe}
An $\eps$-secure SFE protocol must fulfill the following conditions.
\begin{enumerate}[label=(\alph*)]
 \item \label{item:secure-sfe-corr} Correctness: If both players are honest, Alice has input $x \in \mX$ and Bob has input $y \in \mY$, then Bob always receives $f(x,y)$ in the ideal setting. This implies that in an $\eps$-secure protocol, Bob must output a value~$Z$, where
\begin{equation} \label{eq:corr:sfe}
 \Pr[Z = f(x,y)] \geq 1- \eps\;.
\end{equation}

\item \label{item:secure-sfe-security-alice} Security for Alice: Let honest Alice choose her input $X \in \mX$ uniformly at random. Let Bob be malicious. In the ideal setting, the simulator must provide the ideal functionality with a classical input $Y' \in \mY$. It receives the output $Z \in \mZ$ and then outputs a quantum state $\sigma_B$ that may depend on $Y'$ and $Z$. The output of the simulator together with the classical values $X$, $Y'$ and $Z$ defines a state $\sigma_{X B Y'Z}$. This ideal state $\sigma_{XBY'Z}$ must be such that the distribution over all $x$ which are compatible with a given input and output must be uniform with respect to $B$, i.e.,
\begin{align} \label{eq:secA1:sfe}
\sigma_{X B Y'Z}=\sum_{y,z}P_{Y'Z}(y,z)\vecstate{y,z}\kron \sigma^{y,z}_{XB}
\end{align}
with $\sigma^{y,z}_{XB}=\tau_{\{x\in \mX :~f(x,y) = z\}}\kron \sigma^{y,z}_{B}$
and
\begin{align} \label{eq:secA2:sfe}
\dis(\sigma_{XB},\rho_{XB}) \leq \eps\;,
\end{align}
where $\rho_{XB}$ is the resulting state of the real protocol.

\item \label{item:secure-sfe-security-bob} Security for Bob: If Bob is honest and Alice malicious, the simulator outputs a quantum state $\sigma_A$ that is independent of Bob's input $y$. Let $\rho_{A}^y$ be the state that Alice has at the end of the protocol if Bob's input is $y$. The security definition now requires that $\dis(\sigma_A,\rho_{A}^y)\leq \eps$ for all $y \in \mY$. By the triangle inequality, we obtain for any $y,y' \in \mY$
\begin{equation} \label{eq:secB:sfe}
\dis(\rho_{A}^y, \rho_A^{y'}) \leq 2\eps\;.
\end{equation}

\end{enumerate}
\end{condition}

\section{Impossibility of Secure Function Evaluation}\label{sec:imp_sfe}
Our analysis adopts the standard model of quantum two-party computation~\cite{Mayers97,LoChau97,Lo97} (see also~\cite{WTHR11}), where the two players have access to a noiseless quantum and a noiseless classical channel. In each round of the protocol, one party may perform, conditioned on the available classical information, an arbitrary quantum operation on the system in their possession. This operation also generates the input for the available communication channels. By introducing ancillary spaces, the quantum operations of both parties can be purified. If the initial state of the protocol is pure, we can therefore assume that the state at the end of the protocol is pure, conditioned on all classical communication.

The intuition of Lo's original impossibility result for SFE~\cite{Lo97} is that Bob can extract the function values for more than one of his inputs from any implementation of a function that is correct and preserves the privacy of Bob's input. In~\cite{HänggiWinkler2025} is strengthened and made quantitative: We need the following proposition from~\cite{HänggiWinkler2025} providing an explicit upper bound on the probability that Bob can correctly compute the function value for $2 \leq m \leq |\mY|$ of his inputs. For ease of understanding and to make this work self-contained, a proof tailored to our specific setting (Condition 1) is given in Appendix~\ref{app:proof-proposition}.
\begin{restatable}{proposition}{extractionsfe}\label{prop:strong-imposs}
Let $f:\mX \times \mY \rightarrow \mZ$ be a function. For any protocol implementing the function that is correct with probability $1-\eps$ and is $\eps$-secure for Bob in the malicious model, there is an attack that allows a dishonest Bob to compute the values of the function for $m$ of his possible inputs with probability at least $1-2m^2\sqrt{\eps}$.
\end{restatable}
Our first result shows the impossibility to securely implement any \emph{non-trivial} function: A function that has an input $\bar y \in \mY$ that reveals the whole input $x$ of Alice can be securely implemented by a trivial protocol, where Alice sends her input $x$ to Bob, who computes the output $f(x,y)$ from the received value. Since the simulator in the malicious model can always give the input $\bar y$ to the ideal functionality, there is a simulation for every strategy of dishonest Bob. Therefore, we can w.l.o.g.\ consider only non-trivial functions.

\begin{definition}[Non-trivial Functions]
A function $f:\mX \times \mY \rightarrow \mZ$ is \emph{non-trivial} or \emph{partially concealing} if and only if
\begin{align} \label{eq:condition-partially-concealing}
\forall y \in \mY\;  \exists x \neq x' \in \mX: \; f(x,y) = f(x',y)\;.
\end{align}
\end{definition}

\begin{restatable}{theorem}{impossibilitysfe}\label{thm:impossibility-sfe}
For any non-trivial function $f:\mX \times \mY \rightarrow \mZ$, there exists a constant $c_f > 0$ such that no $\eps$-secure quantum protocol in the malicious model exists for $\eps \leq c_f$.
\end{restatable}

\section{Reductions of Secure Function Evaluation}\label{sec:reductions-of-sfe}
We now move to the lower bounds on the minimal number of resources necessary when implementing SFE. We first define the setup of the considered SFE reductions.  To model the resource, the two players are given access to trusted distributed randomness $P_{UV}$. Alice receives $U$ and Bob receives $V$ at the beginning of the protocol. 

Let $X$ and $Y$ be two random variables. The \emph{minimum sufficient statistics} of $X$ with respect to $Y$ is defined as the random variable $\K{X}{Y}:=P_{Y|X=x}$. Since $X$ and $Y$ are independent given $\K{X}{Y}$ (see for example \cite{FiWoWu04}), any protocol with access to a primitive $P_{UV}$ can be transformed into a protocol with access to $P_{\K{U}{V},\K{V}{U}}$ (without compromising the security) because the players can privately compute $P_{UV}$ from $P_{\K{U}{V},\K{V}{U}}$. Thus, we restrict in the following to primitives $P_{UV}$ where $U=\K{U}{V}$ and $V=\K{V}{U}$. 

The classical primitive $P_{UV}$ can be modeled by the quantum primitive \[ \ket{\psi}_{UVE}=\sum_{u,v} \sqrt{P_{UV}(u,v)} \cdot \ket{u,v}_{UV} \otimes \ket{u,v}_E\] that distributes the values $u$ and $v$ to the two players and keeps the environment system~$E$ which holds a purification of the state. This is sufficient to model reductions to $\OT{t}{n}{k}$ and $\RabinOT{p}{k}$ since these primitives are equivalent to trusted randomness according to a certain distribution $P_{UV}$. More precisely, for each of these primitives there exist two protocols: one that generates the distributed randomness using \emph{one} instance of the primitive, and one that implements \emph{one} instance of the primitive using the distributed randomness as input to the two parties~ \cite{BBCS92,Beaver95}. The randomized primitives are obtained by simply choosing all inputs uniformly at random. The corresponding primitive $P_{UV}$, which is also known as an \emph{oblivious key}, is given by the distribution 
\begin{align*}
P_{UV}((x_0,\ldots,x_{n-1}),(c,y))=\begin{cases}
  \frac{1}{n2^{nk}}, & \text{if } y=x_c\;\\
  0, & \text{otherwise }\;,
\end{cases}
\end{align*}
where $(x_0,\ldots,x_{n-1})\in (\sbin^k)^n$ and $(c,y)\in [n]\times \sbin^k$. For a more detailed introduction, we refer to \cite{WW10,winkler:diss}.

By introducing ancillary spaces, the quantum operations of both parties can again be purified. Thus, we can assume that the state at the end of the protocol \emph{including the environment $E$} is pure, conditioned on all classical communication.

\subsection{Lower Bound for Secure Function Evaluation}
We consider secure implementations of functions $f:\mX \times \mY \rightarrow \mZ$. We call an input $x \in \mX$ \emph{redundant} if there exists $x' \neq x \in \mX$ such that $f(x,y) = f(x',y)$ for all $y \in \mY$. We define a restriction $f':\mX\setminus\{x\} \times \mY \rightarrow \mZ$ of $f$ with $f'(x)=f(x)$ for all $x \in \mX\setminus\{x\}$ by removing the redundant input $x$. A secure implementation of $f$ obviously implies a secure implementation of $f'$. The impossibility of securely implementing $f'$, therefore, implies the impossibility to securely implement $f$. We can thus w.l.o.g. consider \emph{non-redundant} functions.
\begin{definition}[Non-redundant Functions]
A function $f:\mX \times \mY \rightarrow \mZ$ is non-redundant if and only if
\begin{align} \label{eq:condition-sfe}
\forall x\neq x' \in \mX\;  \exists y \in \mY: \; f(x,y) \neq f(x',y)\;.
\end{align}
\end{definition}
For all non-redundant functions, it is possible to compute $x$ from the set $\{(f(x,y),y):~y\in\mY\}$ for all $x \in \mX$, i.e., from the set of all function values for a given input of Bob.   

Before stating our main result, we give an outline of its proof: If the implementation of the function is \emph{perfectly} secure for Alice, Bob only learns the function value for his selected input. The fact that the function partially conceals Alice's input, implies a lower bound on the entropy of her input conditioned on Bob's system. 

To derive a bound on the conditional entropies of the initial resource state $\ket{\psi}_{ABE}$, we compare the original protocol to a modified protocol, with the only difference that Bob additionally receives the state of the environment $\rho_E$. This leads to the state between Alice and Bob being pure and corresponds to the case where Alice and Bob implement the function from scratch. In the modified protocol, Bob can compute the function values for all his inputs using Proposition~\ref{prop:strong-imposs}. For any non-trivial function, Bob can, therefore, completely determine Alice's input, and the entropy of her input is \emph{zero} conditioned on Bob's system.

The difference between the high entropy in the secure protocol and the low (zero) entropy in the modified protocol must come from the only difference between the two: Bob's access to the environment state $\rho_E$. However, the dimension of $\rho_E$ limits the possible decrease in entropy. Since the environment $\rho_E$ is a purification of the initial resource $P_{UV}$, this also leads to a lower bound on the size of this trusted distribution and, therefore, for example, the number of OTs it represents.

To extend this reasoning to statistically secure protocols, we take the errors into account and apply known continuity bounds to the relevant conditional entropies.

For any non-redundant function, we can derive the following lower bound for the conditional max-entropies of a resource primitive $P_{UV}$ used to securely implement the function. The proof that follows the above intuition is given in Appendix~\ref{app:main-result}.

\begin{restatable}{theorem}{sfereductions}\label{thm:sfe-reductions}
In any protocol that implements a non-redundant function $f:\mX \times \mY \rightarrow \mZ$ from a primitive $P_{UV}$ with an error $\eps$ in the malicious model, it must hold that
\begin{equation*}
    \Chmaxeps{}{U}{V}+\Chmaxeps{}{V}{U}\geq t - (\eps +\eps') \cdot \log|\mX| - g(\eps)-h(\eps') \;,
\end{equation*}
where $t:=\min_{y \in \mY} \cHS{X}{f(X,y)}$, $g: x \mapsto (1+x)\cdot h\left(\frac{x}{1+x}\right)$ and $\eps':=2|\mY|^2\sqrt{\eps}$.
\end{restatable}

Since any non-trivial function is partially concealing, we have $\min_{y \in \mY} \cHS{X}{f(X,y)} > 0$ for any non-trivial function $f$ (see also proof of Theorem~\ref{thm:impossibility-sfe} in Appendix~\ref{app:impossibility-sfe}).

Theorem~\ref{thm:sfe-reductions} shows that $\min_{y \in \mY} \cHS{X}{f(X,y)}$ is a lower bound on the sum of two conditional max-entropies $\Chmaxeps{}{U}{V}$ and $\Chmaxeps{}{V}{U}$. The classical bound~\cite{WolWul08,PrabhakaranP14} on the other hand, only contains one conditional entropy $\cHS{U}{V}$.  This begs the question whether the two entropies are necessary. This is indeed the case. The reason is that quantum protocols can efficiently reverse $\OT{1}{2}{k}$~\cite{WW14} and violate this classical bound by an arbitrarily large factor. This approach can be generalized to other resource primitives and shows that quantum protocols can take advantage of the entropy contained in $P_{UV}$ in both directions to achieve $\min_{y \in \mY} \cHS{X}{f(X,y)}$ and securely implement a function. 

\subsection{Applications of Lower Bound}

The following statements all directly follow from Theorem~\ref{thm:sfe-reductions}. The proofs are given in Appendix~\ref{app:applications}.

\paragraph{1-out-of-n bit-OT cannot be extended}

In~\cite{WW14} it has been shown that $\OT{1}{2}{1}$ cannot be extended, i.e., there does not exist a protocol using $n$ instances of $\OT{1}{2}{1}$ that securely implements $m > n$ instances of $\OT{1}{2}{1}$. The statement of Theorem~\ref{thm:sfe-reductions} allows us to generalize this result to the case of extensions of $\OT{1}{n}{1}$.
\begin{restatable}{corollary}{extendingot}\label{cor:extend-ot}
For any $n \geq 2$, there exists a constant $c_n > 0$ such that any quantum reduction of $m + 1$ instances of $\OT{1}{n}{1}$ to $m$ instances of $\OT{1}{n}{1}$ must have an error of at least $\frac{c_n}{m}$.
\end{restatable}

\paragraph{The number of choices of OT cannot be extended}

We consider reductions of $\OT{1}{n}{k}$ to $\OT{1}{2}{k}$. A possible generalization of the proof of Theorem 8 in~\cite{WW14} would imply a lower bound of essentially one instance. 

Our Theorem~\ref{thm:sfe-reductions} implies that we need at least roughly $n-1$ instances of $\OT{1}{2}{k}$ for large $k$. Thus, our result can improve the previously known bound by an arbitrarily large factor, depending on $n$.
To obtain this stronger bound, it is crucial that Bob can extract \emph{all} inputs with high probability according to Proposition~\ref{prop:strong-imposs}.

\begin{restatable}{corollary}{extendingchoices}\label{cor:extend-choices-ot}
Any implementation of $\OT{1}{n}{k}$ with an error $\eps \geq 0$ with $m$ instances of $\OT{1}{2}{k}$ as a resource must fulfill
\begin{equation*}
   m \geq \frac{\left((1-\eps-\eps')\cdot n -1\right)k-g(\eps)-h(\eps')}{k+1}\;,
\end{equation*}
where $\eps':=2n^2\sqrt{\eps}$ and $g(x):=(1+x)\cdot h\left(\frac{x}{1+x}\right)$.
\end{restatable}

There is an information-theoretically secure classical protocol that implements $\OT{1}{n}{k}$ from $n-1$ instances of $\OT{1}{2}{k}$~\cite{DodMic99}. Quantum protocols can therefore not improve (substantially) on this.

\paragraph{Inner product function} 

The \emph{inner-product-modulo-two} function $\IP{n}:\{0,1\}^n\times \{0,1\}^n \rightarrow \{0,1\}^n$ is defined as $\IP{n}(x,y) := \oplus_{i=1}^n x_iy_i$. 
Theorem~\ref{thm:sfe-reductions} implies that approximately $n/2$ instances of $\OT{1}{2}{1}$ are needed to compute the inner product function.

\begin{restatable}{corollary}{innerproduct}\label{cor:reductions:inner-product}
For any protocol that implements the inner-product-modulo-two function $\IP{n} : \sbin^n \times \sbin^n \rightarrow \sbin$ from $m$ instances of $\OT{1}{2}{1}$ with an error $\eps \geq 0$ in the \emph{malicious} model, it must hold that 
\begin{equation*}
2m \geq (n-1) - (\eps+\eps') \cdot n - g(\eps) - h(\eps')\;, 
\end{equation*}
where $\eps':=2n^2\sqrt{\eps}$ and $g: x \mapsto (1+x)\cdot h\left(\frac{x}{1+x}\right)$.
\end{restatable}
There exists a perfectly secure classical protocol that computes the inner product function from $n$ instances of $\OT{1}{2}{1}$~\cite{BeiMal04,WW14} and that can be used also in the quantum case. This implies that the bound of Corollary~\ref{cor:reductions:inner-product} is tight up to a factor of two.

\paragraph{Equality function} 

The \emph{equality} function $\EQ{n}:\{0,1\}^n\times \{0,1\}^n\rightarrow \{0,1\}$ is defined as
%$\EQ{n}(x,y)=1$ if $x=y$ and $\EQ{n}(x,y)=0$ otherwise.
\begin{align*}
\EQ{n}(x,y):=\begin{cases}
  1, & \text{if } x=y\;,\\
  0, & \text{otherwise }\;.
\end{cases}
\end{align*}

\begin{restatable}{corollary}{equalityfunction}\label{cor:reductions:equality}
For any protocol that implements the equality function $\EQ{n} : \sbin^n \times \sbin^n \rightarrow \sbin$ from $m$ instances of $\OT{1}{2}{1}$ with an error $\eps$ in the \emph{malicious} model, it must hold for all $0 < k \leq n$ that 
\begin{equation*}
2m \geq (1-\eps-\eps')\cdot k - g(\eps)-h(\eps') -1\;, 
\end{equation*}
where $\eps':=2^{2k+1}\sqrt{\eps}$ and $g: x \mapsto (1+x)\cdot h\left(\frac{x}{1+x}\right)$.
\end{restatable}
There exists a classical protocol that is \emph{also} secure in the quantum case and that implements $\EQ{n}$ from $2(k-1)$ instances of $\OT{1}{2}{1}$~\cite{BeiMal04,WW14} with an error of $2^{-k}$. The bound of Corollary~\ref{cor:reductions:equality} is therefore tight up to a constant factor.
%\section{Conclusion}
%%%%%%
%% Appendix:
%% If needed a single appendix is created by
%%
%\appendix
%%
%% If several appendices are needed, then the command
%%
% \appendices
%%
%% in combination with further \section commands can be used.
%%%%%%

\section*{Acknowledgment}
The authors would like to thank J\"urg Wullschleger for helpful discussions. This work was supported by the Swiss National Science Foundation Practice-to-Science Grant No 199084.

%%%%%%
%% To balance the columns at the last page of the paper use this
%% command:
%%
%\enlargethispage{-1.2cm} 
%%
%% If the balancing should occur in the middle of the references, use
%% the following trigger:
%%
%\IEEEtriggeratref{7}
%%
%% which triggers a \newpage (i.e., new column) just before the given
%% reference number. Note that you need to adapt this if you modify
%% the paper.  The "triggered" command can be changed if desired:
%%
%\IEEEtriggercmd{\enlargethispage{-20cm}}
%%
%%%%%%

%%%%%%
%% References:
%% We recommend the usage of BibTeX:
%%

\bibliographystyle{IEEEtran}
\bibliography{refs}
%%
%% where we here have assumed the existence of the files
%% definitions.bib and bibliofile.bib.
%% BibTeX documentation can be obtained at:
%% http://www.ctan.org/tex-archive/biblio/bibtex/contrib/doc/
%%%%%%
\appendices
\section{Preliminaries}

We will use the following properties of the trace distance (see, for example,~\cite{NieChu00}): The trace-distance is a metric and, therefore, fulfills the triangle inequality. Furthermore, the application of a TP-CPM to two states cannot decrease their distance and the trace distance is preserved under unitary transformations.

A \emph{measurement} $\mE$ can be seen as a TP-CPM that takes a quantum state
on system $A$ and maps it to a classical register $X$, which contains the measurement result, and a system $A'$, which contains the post-measurement state. Such a measurement map is of the form $\mE(\rhoA)=\sum_x \proj{x}_X\otimes M_x \rho_{A} M_x^\dagger$, where all $\ket{x}$ are from an orthonormal basis of $\hX$ and the $M_x$ are such that $\sum_x M_x^{\dagger}M_x=\id_A$. The symbol $\idA$ denotes either the identity operator on $\hi{A}$ or the identity operator on $\posops{\hA}$; it should be clear from the context which one is meant. The probability to obtain the measurement result $x$ is $\tr (  M_x \rho_{A} M_x^\dagger)$. Furthermore, the state of $A'$ after result $x$ is $M_x \rho_{A} M_x^\dagger/\tr (  M_x \rho_{A} M_x^\dagger)$. A \emph{projective} measurement is a special case of a quantum measurement where all measurement operators $M_x$ are orthogonal projectors. If the post-measurement state is not important, and we are only interested in the classical result, the measurement can be fully described by the set of operators $\{ \mE_x\}_x=\{ M_x^\dagger M_x\}_x$.

\subsection{Entropies}
We will make use of the following continuity bound~\cite{AliFan04,Winter2016} for the conditional entropy of any two cq-states $\rho_{XB}, \sigma_{XB} \in \hi{AB}$ with  $\dis(\rho_{XB}, \sigma_{XB}) \leq \eps$
\begin{equation} \label{eq:alifan-winter}
\left|\chvn{X}{B}{\rho} -\chvn{X}{B}{\sigma}\right| \leq \eps\cdot \log |X|+g(\eps)\;,
\end{equation}
where $g:x \mapsto (1+x)\cdot h\left(\frac{x}{1+x}\right)$. The conditional entropy is concave (see for example~\cite{NieChu00}), i.e., it holds that 
\begin{equation}\label{eq:cond-entropy-concavity}
 \chvn{A}{B}{\rho} \geq \sum_{x \in \mX} P_X(x)\chvn{A}{B}{\rho^x}\;,   
\end{equation}
for any probability distribution $P_X$ over $\mX$ and any set of states $\{\rho^x_{AB}\}_{x \in \mX}$ with $\rho_{AB}=\sum_{x \in \mX} P_X(x)\rho^x_{AB}$.
Let $X \in \mX$. If there exists a measurement on $B$ with outcome $X'$ such that $\Pr[X' \neq X] \leq \eps$, then we can apply the data-processing inequality (see for example~\cite{NieChu00}) and Fano's inequality~\cite{Fano61} to get
\begin{equation} \label{eq:fano-q}
  \chvn{X}{B}{\rho} \leq \chvn{X}{X'}{} \leq h(\eps) + \eps \cdot \log |\mX|\;.
\end{equation}
The conditional entropy $\chvn{A}{B}{\rho}$ can decrease by at most $\log|\mZ|$ when conditioning on an additional classical system $Z$, i.e., for any tripartite state $\rhoABZ$ that is classical on $Z$ with respect to some orthonormal basis $\{\ket{z}\}_{z \in \mZ}$, we have (see for example~\cite{Tom2012} for a proof)
\begin{align}\label{chain-rule-c}
 \chvn{A}{BZ}{\rho}\geq \chvn{A}{B}{\rho}-\log |\mZ|\;.
\end{align}
The following lemma shows that the entropy $\chvn{A}{BC}{\rho}$ cannot increase too much when a projective measurement is applied to the system $C$.
\begin{lemma}\cite{WTHR11,Tom2012}\label{lem:data-processing-bound-vn}
  Let $\rhoABC$ be a tripartite state. Furthermore, let $\mE$ be a projective measurement in the basis $\{\ket{z} \}_{z \in \mZ}$ on C and 
  $\idx{\rho}{ABZ} := (\idAB \kron \mE )(\rho_{ABC})$. Then
  \begin{align*}
  	\chvn{A}{BC}{\rho} \geq \chvn{A}{BZ}{\rho} - \log \abs{\mZ} \; .
  \end{align*}
\end{lemma}

\section{Proof of~\cref{prop:strong-imposs}}\label{app:proof-proposition}
In the following, we will show that for any protocol for secure function evaluation that is $\eps$-correct and is $\eps$-secure for Bob, there exists an attack that allows Bob to compute the function value for $m$ of his inputs for any $2 \leq m \leq |\mY|$. The proof will use  Uhlmann's theorem~\cite{Uhlman76} and the Gentle Measurement Lemma~\cite{Winter99,Wilde13}. 

First, we consider two pure states, $\vecstate{\psiAB^{0}}$ and $\vecstate{\psiAB^{1}}$, of Alice and Bob. If the two marginal states, $\rho_A^0$ and $\rho_A^1$, are close in trace distance, then Uhlmann's theorem guarantees the existence of a unitary acting on Alice's side only that transforms the first state $\vecstate{\psiAB^{0}}$ to a state that is close to the second state $\vecstate{\psiAB^{1}}$. The following lemma shows that this also holds when Alice and Bob share two states that are pure conditioned on all the classical information shared between the two parties (classical communication). The proof of the lemma can be found in~\cite{WTHR11}.
\begin{lemma}\label{lem:classical-attack}
For $b \in \{0,1\}$, let 
\[\rho_{XX'AB}^b=\sum_xP_b(x)\vecstate{x}_{X}\kron \vecstate{x}_{X'}\kron \vecstate{\psiAB^{x,b}}\]
with $\dis(\rho_{X'B}^0,\rho_{X'B}^1)\leq \eps$. Then there exists a unitary $U_{AX}$ such that 
\[\dis(\rho_{XX'AB}'^1,\rho_{XX'AB}^1)\leq \sqrt{2 \eps} \]
where $\rho_{XX'AB}'^1=(U_{XA}\kron \id_{X'B})\rho_{XX'AB}^0(U_{XA}\kron \id_{X'B})^\dagger$.
\end{lemma}
The Gentle Measurement Lemma~\cite{Winter99,Wilde13} implies the following upper bound for the disturbance caused by a measurement to a classical-quantum state.
\begin{lemma}\cite{HänggiWinkler2025}\label{lem:gentle-measurement-cq}
Consider a cq-state $\rho_{XA}=\sum_{x \in \mX} P_X(x)\ket{x}\bra{x}\kron \rho_{A}^x$ and measurement acting on system $A$ defined by a trace-preserving completely positive map $\mathcal{E}$ with $\mathcal{E}:\sigma_{XA} \mapsto \sum_x \ket{x}\bra{x}_{X'}\kron\mathcal{E}_x\sigma_{XA}\mathcal{E}_x^\dagger$. Suppose that for each $x \in \mX$ the measurement operator $\mathcal{E}_x$ has a high probability of detecting state $\rho_{A}^x$, i.e., we have $\tr(\mE_x\mathcal\rho_A^x) \geq 1-\eps$ where $\mE_x=M_x^\dagger M_x$. Then for the state $\rho_{X'XA}'=\mE(\rho_{XA})$ we have that $\Pr[X' = X]\geq 1-\eps$ and the post-measurement state, ignoring the measurement outcome, 
\begin{equation*}
\rho_{XA}'=\sum_x(\id_X \kron M_x)\rho_{XA}(\id_X\kron M_x)^\dagger
\end{equation*}is $(\sqrt{\eps}+\eps)$-close to the original state $\rho$, i.e., 
\begin{equation*}
\dis(\rho_{XA},\rho_{XA}') \leq \sqrt{\eps}+\eps\;.
\end{equation*}
\end{lemma}

Using Lemma~\ref{lem:gentle-measurement-cq} we can derive the following lower bound on the success probability of a generic attack of Bob on any SFE protocol.
\extractionsfe*
\begin{proof}
Let $\mY = \{y_1,\ldots,y_n\}$. We assume that Alice chooses her input $X$ uniformly at random from $\mX$. For all inputs $y \in \mY$ of Bob, we can assume that the  corresponding state $\rho_{AB}^y$ at the end of the protocol is pure given all classical communication.

The correctness of the protocol implies that Bob can compute $f(X,y)$ from $\rho^y_{AB}$ with probability at least $1-\eps$. From Lemma~\ref{lem:gentle-measurement-cq} we know that the resulting state $\Tilde{\rho}^y_{AB}$ after computing $f(X,y)$ is $(\sqrt{\eps}+\eps)$-close to $\rho^y_{AB}$, i.e., 
\begin{equation}\label{eq:damage-measurement}
\dis(\rho_{AB}^y, \Tilde{\rho}_{AB}^y) \leq \sqrt{\eps}+\eps\;.
\end{equation}
We define $\rho_{AB}^i:= \rho_{AB}^{y_i}$ for all $i \in \{1,\ldots , n\}$. Since the protocol is secure for an honest Bob, we have that for all $i,j \in \{1,\dots,n\}$
\begin{equation*}
\dis(\rho_{A}^i, \rho_A^j) \leq 2\eps\;.
\end{equation*}
By Lemma \ref{lem:classical-attack} this implies that there exists a unitary $U_{B,i,j}$ such that
\begin{equation}\label{ineq:imposs:ot-full:uhlmann}
\dis\big(\rho_{AB}'^j, \rho_{AB}^j\big)\leq 2\sqrt{\eps}\;,
\end{equation}
where $\rho_{AB}'^j= (\id \kron U_{B,i,j})\rho_{AB}^i(\id \kron U_{B,i,j})^\dagger$. Let $\mathcal{E}_i$ be the operation that computes the value $f(X,y_i)$ from Bob's system $B$ and let $\Tilde{\rho}_{AB}^i = \mathcal{E}_i(\rho_{AB}^i)$ be the resulting state. 
Then, omitting the identity on $A$, it holds that
\begin{align*}
\begin{split}  \dis(U_{B,i,j}\mathcal{E}_i(\rho_{AB}^i)U_{B,i,j}^\dagger &, \rho_{AB}^j)\\
       &=  \dis(U_{B,i,j}\Tilde{\rho}_{AB}^iU_{B,i,j}^\dagger, \rho_{AB}^j)\\
      &\leq  \dis(U_{B,i,j}\Tilde{\rho}_{AB}^iU_{B,i,j}^\dagger, U_{B,i,j}\rho_{AB}^iU_{B,i,j}^\dagger)\\&~~~~~~~~~~~~~+ \dis(U_{B,i,j}\rho_{AB}^iU_{B,i,j}^\dagger, \rho_{AB}^j)\\
      &= \dis( \Tilde{\rho}_{AB}^i, \rho_{AB}^i) + \dis\big(\rho_{AB}'^j, \rho_{AB}^j\big)\\
      &\leq (\sqrt{\eps}+\eps)+2\sqrt{\eps} \\
      &=  3\sqrt{\eps}+\eps\;,
\end{split}
\end{align*}
where the first inequality follows from the triangle inequality, the next equality from  the fact that the trace distance is preserved under unitary transformations and the second inequality follows from~\eqref{eq:damage-measurement} and~\eqref{ineq:imposs:ot-full:uhlmann}.

We have shown that for all $i,j \in \{1,\dots,n\}$ there exists quantum operation $S_{i,j}$, namely first measuring the value $X_i$ and then applying the unitary $U_{B,i,j}$, that computes the value $X_i$ from the state $\rho_{AB}^i$ and results in a state that is $(3\sqrt{\eps}+\eps)$-close to $\rho_{AB}^j$, i.e,
\begin{equation}\label{eq:damage-operation}
    \dis(S_{i,j}(\rho_{AB}^i), \rho_{AB}^j) \leq 3\sqrt{\eps}+\eps\;.
\end{equation}
Thus, the attack allows Bob to compute all of Alice's inputs: Bob has input $y_1$ and honestly follows the protocol except from keeping his state purified. At the end of the protocol, he computes successively all values $f(X,y_j)$ using the above operation $S_{i,i+1}$, which means that  he computes the next value $f(X,y_i)$ and then applies a unitary to rotate the state close to the next target state $\rho_{AB}^{i+1}$. Let $S_i:=S_{i-1,i}$. Next, we will show that 
\begin{equation}\label{eq:damage-succession}
    \dis(S_l(S_{l-1}(\ldots S_2(\rho_{AB}^1))), \rho_{AB}^l)\leq (l-1)(3\sqrt{\eps}+\eps)\;.
\end{equation}
From~\eqref{eq:damage-operation}, we know that the statement holds for $l=2$ 
\begin{equation*}
    \dis(S_2(\rho_{AB}^1),\rho_{AB}^2)\leq 3\sqrt{\eps}+\eps\;.
\end{equation*}
Assume that the statement holds for $l-1$, then we have for any $2 \leq l \leq n$
\begin{align*}
\begin{split}
    \dis(S_l&(S_{l-1}(\ldots S_2(\rho_{AB}^1))), \rho_{AB}^l)\\&\leq\dis((S_l(S_{l-1}(\ldots S_2(\rho_{AB}^1))), S_l(\rho_{AB}^{l-1}))+\dis( S_l(\rho_{AB}^{l-1}),\rho_{AB}^l)\\
    &\leq\dis((S_{l-1}(\ldots S_2(\rho_{AB}^1)), \rho_{AB}^{l-1})+\dis( S_l(\rho_{AB}^{l-1}),\rho_{AB}^l)\\
    &\leq (l-2)(3\sqrt{\eps}+\eps) + 3\sqrt{\eps}+ \eps\\
    &= (l-1))(3\sqrt{\eps}+\eps)\;,
    \end{split}
\end{align*}
where we first applied the triangle inequality, then we used the fact that quantum operations cannot increase the trace distance in the second inequality, and finally we used~\eqref{eq:damage-operation} and the assumption that~\eqref{eq:damage-succession} holds for $l-1$ in the third inequality.

Let $\eps_l$  be the probability that Bob fails to compute the $l$-th value correctly. We know that Bob can compute the value $f(X,y_j)$ from the state $\rho_{AB}^j$ with probability at least $1-\eps$ and that $\eps_1 \leq \eps$. We can upper bound the probability for the other $\eps_l$ adding the distance from the state $\rho_{AB}^l$ given by~\eqref{eq:damage-succession} and the error $\eps$, i.e.,
\begin{align}\label{ineq:prob-fail-step}
    \eps_l &\leq (l-1)(3\sqrt{\eps}+\eps)+\eps
    \leq l\cdot(3\sqrt{\eps}+\eps)\;
\end{align}
Thus, we can apply the union bound to obtain the claimed upper bound for the probability that Bob fails to compute any of the $m \geq 2$ values correctly
\begin{align*}
\sum_{l=1}^{m}\eps_l&\leq \eps_1 +\sum_{l=2}^{m}l\cdot(3\sqrt{\eps}+\eps)\\
&\leq\eps+(3\sqrt{\eps}+\eps)\sum_{l=2}^{m}l\\
&=\eps+(3\sqrt{\eps}+\eps)(m(m+1)/2-1)\\
&=\frac{1}{2}(3\sqrt{\eps}+\eps)(m^2+m)-3\sqrt{\eps}\;.
\end{align*}
We can in the following assume that $\eps \leq \frac{1}{64}$. Otherwise $1-2m^2\sqrt{\eps}\leq 0$ for $m\geq 2$ and the statement of the theorem is trivially true. Let $m \geq 4$. Since $\eps \leq \sqrt{\eps}/8$ and $m \leq m^2/4$, it holds that 
\begin{align*}
\frac{1}{2}(3\sqrt{\eps}+\eps)(m^2+m)-3\sqrt{\eps} &\leq \frac{1}{2}(3\sqrt{\eps}+\sqrt{\eps}/8)(m^2+m^2/4)\\
&=\tfrac{125}{64}\cdot m^2\sqrt{\eps}\\
&\leq 2m^2\sqrt{\eps}\;.
\end{align*}
It can easily be checked that $\frac{1}{2}(3\sqrt{\eps}+\eps)(m^2+m)-3\sqrt{\eps}\leq 2m^2\sqrt{\eps}$ for $m\in \{2,3\}$. This implies the claimed statement for all $2 \leq m \leq n$.
\end{proof}

\section{Proof of~\cref{thm:impossibility-sfe}}\label{app:impossibility-sfe}
\impossibilitysfe*
\begin{proof}
Let Alice choose her input uniformly at random from $\mX$. Let $\rho_{XB}$ be the real state at the end of the protocol and $\sigma_{X B Y'Z}$ the ideal state given by~\eqref{eq:secA1:sfe}, i.e.,
\begin{align*} 
\sigma_{X B Y'Z}=\sum_{y,z}P_{Y'Z}(y,z)\vecstate{y,z}\kron \sigma^{y,z}_{XB}
\end{align*}
with $\sigma^{y,z}_{XB}=\tau_{\{x\in \mX :~f(x,y) = z\}}\kron \sigma^{y,z}_{B}$. Since the function is non-trivial and, therefore, partially concealing, there exist for any input $y \in \mY$ of Bob two inputs $x,x' \in \mX$ such that $f(x,y)=f(x',y)$. If Bob has input $y$ and receives output $z:=f(x,y) \in \mZ$, he can therefore guess Alice's input with probability at most $\tfrac{1}{2}$. Therefore, the probability to guess Alice's input from the ideal state $\sigma_{X B Y'Z}$ can be at most
\begin{align*}
    \tfrac{|\mX|-2}{|\mX|} + \tfrac{1}{2}\cdot \tfrac{2}{|\mX|}= 1- \tfrac{1}{|\mX|}\;
\end{align*}
for any strategy of a dishonest Bob (simulator). The ideal state is $\eps$-close to the real state according to~\eqref{eq:secA2:sfe}. Bob can therefore guess Alice's input with probability at most  $1-\tfrac{1}{|\mX|} + \eps$ from the real state at the end of the protocol. Proposition~\ref{prop:strong-imposs} implies that a dishonest Bob can compute the function values for all his possible inputs with probability at least $1-2|\mY|^2\sqrt{\eps}$. This implies that for any $\eps$-secure protocol it must hold that
\begin{align*}
    1-2|\mY|^2\sqrt{\eps} \leq 1- \tfrac{1}{|\mX|} + \eps\;.
\end{align*}
Solving the inequality for $\eps$ implies the claimed statement.
\end{proof}

\section{Proof of~\cref{thm:sfe-reductions}}\label{app:main-result}
\sfereductions*

\begin{proof}
Let Alice's input $X$ be uniformly distributed over $\mX$ and let Bob's input be fixed to $\bar y$. Let the final state of the protocol on Alice's and Bob's system be $\rho_{AB}^{\bar y}$, when both players are honest. Since the protocol is $\eps$-secure for Alice in the malicious model, we can conclude from~\eqref{eq:secA1:sfe} that there exists an ideal state $\sigma_{XB}$, for which the distribution over all $x$ that are compatible with a given output must be uniform with respect to $B$. According to~\eqref{eq:secA2:sfe} the ideal state is $\eps$-close to the state $\rho_{XB}$. Since the conditional entropy is concave, it holds that
\begin{align} \label{eq:lower-bound-entropy}
\begin{split}
\chvn{X}{B}{\sigma}&\geq \sum_{y,z}P_{Y'Z}(y,z)\chvn{X}{B}{\sigma^{y,z}}\\&\geq \min_{y \in \mY} \cHS{X}{f(X,y)}=t\;.
\end{split}
\end{align}

Then we can derive the following lower bound on the uncertainty that Bob has about Alice's input conditioned on his state after the execution of the real protocol by using the continuity bound for the conditional entropy of cq-states~\eqref{eq:alifan-winter}
\begin{equation}\label{eq:cond-ent:lower-bound:sfe}
\begin{split}
   \chvn{X}{B}{\rho^{\bar y}} &\geq \min_{y \in \mY} \cHS{X}{f(X,y)} -\eps \cdot \log |\mX| - g(\eps)\;.
\end{split}
\end{equation}

We consider a modified protocol that starts from a state 
\begin{align*}
 \ket{\psi}_{UVB'}=\sum_{u,v} \sqrt{P_{UV}(u,v)} \cdot \ket{u,v}_{UV} \kron \ket{u,v}_{B'}\;,
\end{align*}
where the systems $V$ and $B'$ belong to Bob. Since Bob can simulate the original protocol from the modified protocol, any successful attack of Alice against the modified protocol does obviously imply a successful attack against the original protocol. Let $\rho_{XABB'}$ be the state at the end of the modified protocol. Its marginal state, $\rho_{XAB}$ , is the corresponding state at the end of the original protocol. Furthermore, the state conditioned on the classical communication is again pure. From Proposition~\ref{prop:strong-imposs} we know that Bob can compute $\{f(x,y):y \in \mY\}$ from the state $\rho_{ABB'}^{\bar y}$ with probability at least $\eps':=3|\mY|^2\sqrt{\eps}$. Condition~\eqref{eq:condition-sfe} guarantees that Bob can derive Alice's input from the computed values. Then we can apply equality~\eqref{eq:fano-q} to obtain the upper bound
\begin{align}\label{eq:cond-ent:upper-bound:sfe}
\HS{X}{BB'}_{\rho^{\bar y}} \leq  h(\eps') + \eps' \cdot \log |\mX|\;.
\end{align}
Therefore, it follows from the two inequalities~\eqref{eq:cond-ent:lower-bound:sfe} and~\eqref{eq:cond-ent:upper-bound:sfe} that
\begin{equation*}
 \begin{split}
    \cHS{X}{B}_{\rho^{\bar y}}&-\cHS{X}{BB'}_{\rho^{\bar y}}\\ &\geq t -\eps \cdot \log |\mX| - g(\eps) - h(\eps') - \eps' \cdot \log |\mX|\\
    &=t- (\eps +\eps') \cdot \log|\mX| - g(\eps)-h(\eps') \;.
\end{split}
\end{equation*}
Then we can apply inequality~\eqref{chain-rule-c} and Lemma~\ref{lem:data-processing-bound-vn} to obtain the claimed lower bound for the conditional max-entropies
\begin{equation*}
\begin{split}
    \Chmaxeps{}{U}{V}&+\Chmaxeps{}{V}{U}\\&=\max_v\log|\support( P_{U|V=v})|
+\max_u\log|\support(P_{V|U=u})|\\&\geq\cHS{X}{B}_{\rho^{\bar y}}-\cHS{X}{BB'}_{\rho^{\bar y}} \\&\geq t - (\eps +\eps') \cdot \log|\mX| - g(\eps)-h(\eps')\;. 
\end{split} 
\end{equation*}
\end{proof}

\section{Proofs of Applications}\label{app:applications}

\extendingot*
\begin{proof}
We use the lower bound of Theorem~\ref{thm:sfe-reductions} and apply the same reasoning as in~\cite{WW14}: We assume that there exists a composable protocol that implements $m + 1$ instances of $\OT{1}{n}{1}$ from $m$ instances of $\OT{1}{n}{1}$ with an error $\eps$. We can apply this protocol iteratively and implement $4m$ instances of $\OT{1}{n}{1}$ with an error of $\tilde \eps:=3m\eps$, where we assume that Bob follows the protocol. Theorem~\ref{thm:sfe-reductions} implies that 
\begin{equation*}
 m((n-1)+ \log(n))\geq (n-1)4m-(\tilde \eps+\tilde \eps')\cdot 4mn - g(\tilde \eps)-h(\tilde \eps')\;
\end{equation*}
and, therefore,
\begin{align}\label{eq:ext-ot-lower-bound}
\frac{g(\tilde \eps)+h(\tilde \eps')}{m}+(\tilde \eps+\tilde \eps')4n \geq 3(n-1)-\log(n)\;,
\end{align}
where we used the fact that $\Chmaxeps{}{U}{V}+\Chmaxeps{}{V}{U} = m((n-1)+ log(n))$ for the resource $P_{UV}$ that corresponds to $m$ instances of $\OT{1}{n}{1}$ and $\tilde \eps':=3n^2\sqrt{\tilde \eps}$. If $\tilde \eps \leq 1/\left(9\cdot3n^2\right)^2$, then $g(\tilde \eps) + h(\tilde \eps') < 1$ and $\tilde \eps + \tilde \eps' < 1/8$. In this case, it follows from inequality~\eqref{eq:ext-ot-lower-bound} that 
\begin{align*}
1+n/2 > 3(n-1)-\log(n)\;,
\end{align*}
which implies $\log(n) > \frac{5n}{2}-4$. Thus, it must hold that $\tilde \eps \geq 1/\left(9\cdot3n^2\right)^2$ and, therefore, $\eps \geq \frac{c_n}{m}$ with 
\begin{equation*}
c_n:= \frac{1}{3m\left(9\cdot3n^2\right)^2}\;.
\end{equation*}
Note that the above analysis is not optimized to be as tight as possible, but sufficiently tight to prove the existence of the claimed constant $c_n$.
\end{proof}

\extendingchoices*
\begin{proof}
$\Chmaxeps{}{U}{V}+\Chmaxeps{}{V}{U} = m(k+1)$ for the resource $P_{UV}$ that corresponds to $m$ instances of $\OT{1}{2}{k}$. We can conclude from Theorem~\ref{thm:sfe-reductions} that 
\begin{equation*}
\begin{split}
   m(k+1) &\geq  (n-1)\cdot k - (\eps+\eps')\cdot nk -g(\eps)-h(\eps')\\
   &=\left((1-\eps-\eps')\cdot n-1\right)k-g(\eps)-h(\eps')\;,
\end{split}
\end{equation*}
which implies the claimed statement.
\end{proof}

\innerproduct*
\begin{proof}
$\Chmaxeps{}{U}{V}+\Chmaxeps{}{V}{U} = 2m$ for the resource $P_{UV}$ that corresponds to $m$ instances of $\OT{1}{2}{1}$. Let $e_i \in \{0,1\}^n$ be the string that has a one at the $i$-th position and is zero otherwise. Let $\mS:=\{e_i:~1\leq i \leq n\}$. Then the restriction of the inner-product function to inputs from $\{0,1\}^n \times \mS$ satisfies condition~\eqref{eq:condition-sfe} with $|S|=n$. Furthermore, it holds for all $z \in \sbin$ and all $y \in \mY$ that $|\{x \in \mX: f(x, y)=z\}| \geq 2^{n-1}$. Thus, we can apply Theorem~\ref{thm:sfe-reductions} to obtain the claimed lower bound.
\end{proof}

\equalityfunction*
\begin{proof}
$\Chmaxeps{}{U}{V}+\Chmaxeps{}{V}{U} = 2m$ for the resource $P_{UV}$ that corresponds to $m$ instances of $\OT{1}{2}{1}$. We can restrict the input domains of both players to the same subset of $\sbin^n$ of size $2^k$. The restricted function is still non-redundant and it holds that $\min_{y \in \mY} \cHS{X}{f(X,y)} =  (1-2^{-k})\cdot k$. Thus, Theorem~\ref{thm:sfe-reductions} implies that
\begin{align*}
    2m&\geq  (1-2^{-k})\cdot k - (\eps +\eps') \cdot \log|\mX| - g(\eps)-h(\eps')\\
   & \geq (1-\eps-\eps')\cdot k - g(\eps)-h(\eps') -1\;,
\end{align*}
where $\eps':=2^{2k+1}\sqrt{\eps}$ and $g: x \mapsto (1+x)\cdot h\left(\frac{x}{1+x}\right)$.
\end{proof}

%% Or you use manual references (pay attention to consistency and the
%% formatting style!):

\end{document}